\documentclass{amsart}

\usepackage{setspace}
\usepackage{amsmath,amsthm}
\usepackage{amssymb, amsfonts}

\newcommand{\Tr}{\operatornamewithlimits{Tr}}
\newcommand{\dom}{\operatornamewithlimits{dom}}

\newcommand{\Span}{\operatornamewithlimits{Span}}

\newtheorem{theorem}{Theorem}[section]

\newtheorem{proposition}[theorem]{Proposition}
\newtheorem{definition}[theorem]{Definition}
\newtheorem{assumption}[theorem]{Assumption}

\newtheorem{example}[theorem]{Example}

\numberwithin{equation}{section}

\title{Optimal portfolios in commodity futures markets}
\author{Fred Espen Benth}
\author{Jukka Lempa}
\thanks{\emph{Address.} Centre of Mathematics for Applications, University of Oslo, PO Box 1053 Blindern, NO -- 0316 Oslo, e-mail: \texttt{fredb@math.uio.no}, \texttt{jlempa@cma.uio.no}}

\begin{document}

\maketitle

\begin{abstract}
We consider portfolio optimization in futures markets. We model the entire futures price curve at once as a solution of a stochastic partial differential equation.  The agents objective is to maximize her utility from the final wealth when investing in futures contracts. We study a class of futures price curve models which admit a finite-dimensional realization. Using this, we recast the portfolio optimization problem as a finite-dimensional control problem and study its solvability.
\end{abstract}

\section{Introduction}


Futures contracts, see, e.g., \cite{Cox Ingersoll Ross 81}, form an important class of financial instruments exchanged on commodity markets. These contracts convey the right to purchase or sell a specified quantity of the commodity at a fixed price on a fixed future date. This right has to be exercised if the contract is held until the maturity. The price fixed in contract is called the futures price and it is set such that no money is paid upfront, i.e., the initial value of a futures contract is zero. However, a futures contract yields a cash flow during its life time generated by the changes in futures price over time. More precisely, the party in whose favor the futures price change occurs must immediately be paid the full amount of the change by the losing party. Typically, a margin account is set up for this purpose.


The objective of this paper is to study portfolio optimization on commodity futures markets using tools from mathematical finance and stochastic calculus. We assume that the futures prices are settled continuously in time and that there is a liquid market for futures contracts for every time-to-delivery $y>0$. Now we are facing a portfolio optimization problem in infinite dimensions, since we have a stochastic variable, that is, a futures price, for a continuum of times to delivery. To tackle this problem, we start with reparametrized futures price dynamics \'{a}-la Musiela which allows to formulate the dynamics as a solution of a stochastic partial differential equation. This equation can be regarded as a stochastic differential equation taking values in a real separable Hilbert space.  After this, we focus on futures prices models which admit a finite-dimensional realization. Roughly speaking, for a given instantaneous futures price model $t\mapsto f_t$, we look for a sufficiently regular function $\phi:\mathbf{R}_+\times \mathbf{R}^d\rightarrow \mathbf{R}$ and a $d$-dimensional diffusion $Z$ such that $f_t(y)=\phi(y,Z_t)$ for all $t\geq0$. Here, $d$ is the dimension of the realization. Various classes of models admitting finite-dimensional realization have been studied over the recent years, mostly in connection with term structure models of interest rates. In a series of papers including \cite{Bj�rk Gombani, Bj�rk Svensson, Bj�rk Landen, Bj�rk Blix Landen} the authors study geometric aspects and finite-dimensional realizations of interest rate models using differential geometry and systems and control theoretic methods. In another series of papers including \cite{Filipovic Teichmann 1, Filipovic Teichmann 2}, see also \cite{Filipovic}, the authors analyze the properties of finite-dimensional realizations using the so-called convenient analysis on Fr\'{e}chet space. For example, they prove that the function $\phi$ introduced above is necessary affine. In these papers, they provide a general theory on invariance of manifolds for infinite-dimensional stochastic equations and study the term structure models of interest rates as an application. While being mathematically very general, their analysis is quite demanding to the reader. Recently, a more direct approach to affine realizations of term structure models was introduced in \cite{Tappe} -- we follow this approach when studying our futures price model. We also refer to \cite{Ekeland Taflin} and \cite{Ringer Tehranchi}, where portfolio optimization with an infinite dimensional state variable is studied. These studies are concerned with optimal bond portfolios where the underlying interest rate term structure follows a solution of a stochastic partial differential equation. In addition to addressing a whole different application, these studies operate on the infinite dimensional level, whereas we study models with a finite-dimensional realization.

Our paper makes two main contributions. First, we provide a general mathematical framework to treat futures portfolio optimization using the up-to-date theory of the realizations of stochastic partial differential equations. More precisely, we provide conditions under which a given infinite-dimensional portfolio optimization problem can be solved in terms of a finite-dimensional control problem. Here, the finite-dimensional realization (more precisely, the driving coordinate process $Z$) of the futures price curve $t\mapsto f_t$ plays a key role. Furthermore, we discuss an economic interpretation of the coordinate process and how a solution of the finite-dimensional control problem can be connected to the coodinate process and, consequently, back to the infinite-dimensional portfolio problem. Second, we extend some of the results in \cite{Tappe} to a new class of stochastic differential equations. In \cite{Tappe}, the author considers finite-dimensional realizations for HJM term structure models, see \cite{Heath Jarrow Morton}. This leads to the well known HJM drift condition for the underlying dynamics whereas we work with a different drift condition resulting into a different SPDE.

The reminder of the paper is organized as follows. In Section 2 we set up a general framework for modeling of the futures curve. In Section 3 we analyze the case of finite-dimensional realization. The portfolio problem is recast as a finite-dimensional control problem in Section 4.

\section{The Portfolio Problem}

\subsection{Futures Price Dynamics} We start the analysis by defining the function space on which the futures price curve $t\mapsto f_t$ evolves. We use the Musiela parametrization and write the futures curve as a function of time-to-maturity $y\in\mathbf{R}_+$. Following \cite{Filipovic}, see also \cite{Tappe}, we fix a parameter $\alpha>0$ and denote as $H_\alpha$ the space of all absolutely continuous functions $h:\mathbf{R}_+\rightarrow\mathbf{R}$ such that
\[ \|h\|_\alpha:=\left(|h(0)|^2+\int_0^\infty e^{\alpha y} |h'(y)|^2 dy \right)^{\frac{1}{2}}<\infty. \]
Here, the derivative $h'$ is understood in the weak sense. The space $(H_\alpha,\|\cdot\|_\alpha)$ is a separable Hilbert space for which the point evaluation $h\mapsto \delta_y(h): H_\alpha\rightarrow\mathbf{R}$ is a continuous linear functional for each $y\in\mathbf{R}_+$ -- see, e.g., \cite{Tappe}, Theorem 4.1. We observe from the definition of the norm $\|\cdot\|_\alpha$ that the derivative of the functions in $H_\alpha$ decay at exponential rate whereas the actual futures price can very well be non-zero for large times to maturity.

We denote as $(S_t)_{t\geq0}$ the semigroup of right shifts defined as $S_t f(y)=f(y+t)$ for any given $f\in H_\alpha$. Furthermore, denote the differential operator $\frac{\partial}{\partial y}$ as $A$. Then we know that the semigroup $(S_t)_{t\geq0}$ is strongly continuous with infinitesimal generator $A$, see \cite{Tappe}, Thrm. 4.1. The domain of operator $A$ is defined as $\mathcal{D}(A)=\{ h\in H_\alpha \, | \, h'\in H_\alpha \}$. Furthermore, we define domains $\mathcal{D}(A^n):=\left\{ h\in\mathcal{D}(A^{n-1}) \, | \, A^{n-1} h \in \mathcal{D}(A^{n-1}) \right\}$, along with the intersection $\mathcal{D}(A^\infty)=\bigcap_{i=1}^\infty \mathcal{D}(A^i)$. Finally, an element $h\in\mathcal{D}(A^\infty)$ is called \emph{quasi-exponential} if the linear space spanned by the family $\{A^n h\}_{n=1}^\infty$ is finite dimensional.

Let $U$ be a real separable Hilbert space. We assume that $\mathcal{W}$ be a Wiener process defined on a complete filtered probability space $(\Omega,\mathcal{F},\mathbf{F},\mathbf{Q})$, where the filtration $\mathbf{F}=\{\mathcal{F}_t\}_{t\geq0}$ satisfies the usual conditions, and taking values in $U$. Following the HJM-approach, we assume that $\mathbf{Q}$ is a (local) martingale measure. Denote the covariance operator of $\mathcal{W}$ as $Q$ and the associated eigenvectors and -values as $\{e_i\}_{i\in\mathcal{I}}$ and $\{\lambda_i\}_{i\in\mathcal{I}}$, respectively. Here, $\mathcal{I}$ is a countable index set determined by the dimensionality of $U$.  We assume that $\Tr Q=\sum \lambda_i<\infty$. The family $\{e_i\}_{i\in\mathcal{I}}$ is an orthonormal basis for the space $U$ and we can express $\mathcal{W}$ as
\begin{equation}\label{U-valued Wiener process representation}
\mathcal{W}_t=\sum_{i\in\mathcal{I}}\sqrt{\lambda_i}\,\mathcal{W}^i_t\,e_i,
\end{equation}
where $\mathcal{W}^i$ are scalar Wiener processes. Furthermore, we denote $U^Q=Q^{\frac{1}{2}}(U)$, where $Q^{\frac{1}{2}}$ is the pseudo-square root of the covariance operator $Q$. To fix notation, we denote the space of all Hilbert-Schmidt operators from a given Hilbert space $H^1$ to another Hilbert space $H^2$ as $L_{HS}(H^1,H^2)$.

We assume that the risk-neutral futures price dynamics $t\mapsto f_t$ are defined as the solution of the stochastic partial differential equation
\begin{equation}\label{Q-forward}
\begin{split}
df_t(\cdot) =Af_t(\cdot)\,dt+\Sigma(f_t(\cdot))\,d\mathcal{W}_t
        =Af_t(\cdot)\,dt+\sum_{i\in\mathcal{I}}\sqrt{\lambda_i}\,\sigma_i(f_t(\cdot))\,d\mathcal{W}^i_t, \quad f_0\in H_\alpha,
\end{split}
\end{equation}
where, for brevity, $\sigma_i(h(\cdot)):=\Sigma_i(h(\cdot))(e_i)\in H_\alpha$. Here, the mapping $\Sigma:=(\Sigma_i)_{i\in\mathcal{I}}: H_\alpha\rightarrow L_{HS}(U^Q,H_\alpha)$ satisfies appropriate measurability conditions, is Lipschitz and admits the uniform linear growth condition $\|\Sigma(h)\|^2_{L_{HS}(U^Q,H_\alpha)}\leq C(1+\|h\|_\alpha)$, for all $h$, where $C$ does not depend on $h$. These conditions guarantee the existence of a unique mild solution, see, e.g., \cite{DePrato Zabczyk}, Chapter 7. The mild solution of \eqref{Q-forward} can be expressed as
\begin{align*}
 f_t(\cdot)&=S_tf_0(\cdot)+\int_0^t S_{t-u}\,\left(\Sigma(f_u(\cdot))\, d\mathcal{W}_u\right) \\
&=S_tf_0(\cdot)+\sum_{i\in\mathcal{I}}\sqrt{\lambda_i}\int_0^t S_{t-u}\,\left(\sigma_i(f_u(\cdot))\right)\, d\mathcal{W}^i_u.
\end{align*}
Furthermore, we assume that
\[ \mathbf{E}\left[\int_0^T \| \Sigma(f_t(\cdot))\|^2_{L_{HS}(U^Q,H_\alpha)}dt \right]<\infty. \]
This ensures that the mild solution is also a weak solution, see \cite{Gawarecki Mandrekar}, Thrm 3.2.

We need also a description of the futures dynamics under the market measure $\mathbf{P}$. To this end we apply Girsanov's theorem, see, e.g., \cite{DePrato Zabczyk}, Theorem 10.14. For a given $\psi\in U^Q$, define the process $\hat{\mathcal{W}}$ as
\[ d\hat{\mathcal{W}}_t=d\mathcal{W}_t-\psi dt. \]
Then $\hat{\mathcal{W}}$ is a $U$-valued Wiener process under an equivalent measure $\mathbf{P}$ defined via the Radon-Nikodym derivative $\frac{d\mathbf{P}}{d\mathbf{Q}}=\mathcal{E}(-\psi\cdot \mathcal{W})_T$. Here, $\mathcal{E}$ denotes the Dol\'{e}ans-Dade exponential, see \cite{Filipovic}, Chapter II. Furthermore, the covariance operator of $\hat{\mathcal{W}}$ is $Q$. The futures price dynamics can be written under the measure $\mathbf{P}$ as
\begin{align}
\label{P-forward}
df_t(\cdot)&=\left(A f_t(\cdot)+\Sigma(f_t(\cdot))(\psi)\right)\,dt+\Sigma(f_t(\cdot))\,d\hat{\mathcal{W}}_t \nonumber\\
            &:=\nu(f_t(\cdot))dt+\Sigma(f_t(\cdot))\,d\hat{\mathcal{W}}_t.
\end{align}
In financial terms, we can interpret the element $\psi$ as the market price of risk. Being an element of a Hilbert space, the market price of risk $\psi$ can have essentially richer structure than its finite-dimensional counterpart, a constant vector. Indeed, if the space $U^Q$ is a space of functions of the time to maturity $y$, the market price of risk will also depend on time to maturity. Furthermore, we can interpret the term $\Sigma(f_t(\cdot))(\psi)$ as the risk premium. We observe immediately that the risk premium is proportional to the volatility of the futures price curve. Our stochastic dynamics allow for a very flexible modelling of the risk premium, taking into account possible idiosyncraticites between futures contracts with different maturities. Since the $\mathbf{P}$-dynamics of the futures price curve is completely determined by the volatility $\Sigma$ and the market price of risk $\psi$, we call the pair $(\Sigma,\psi)$ a \emph{futures price model}.

\subsection{Portfolios of Futures Contracts} As was mentioned already in the introduction, a futures contract is a derivative security written on the  futures price $f_t$. The initial market value of this contract is zero and it yields a cash flow during its lifetime generated by the fluctuations of the futures price. The profits (losses) caused by the fluctuations are put in (drawn from) a margin account such that the brokers collateral remains on the required level over time. We assume that the investor can use the margin account also as a borrowing account with an interest rate same as the constant risk free rate of return. In order to describe the time evolution of the wealth generated by trading on futures contracts, consider first a single futures contract in discrete time. Denote the length of the time step as $\Delta t$. The parties enter a single futures contract at time $t$ with the delivery at time $T>t+\Delta t$. This position is canceled at time $t+\Delta t$ and a new contract is entered with the same time of delivery -- this is the \emph{marking to market}-procedure. The resulting profit/loss is the associated fluctuation of the futures price, i.e.
\begin{equation*}
\begin{split}
f_{t+\Delta t}(y-\Delta t)-f_{t}(y)=f_{t+\Delta t}(y)-f_{t}(y)-(f_{t+\Delta t}(y)-f_{t+\Delta t}(y-\Delta t)),
\end{split}
\end{equation*}
where $y=T-t$ is the time to delivery at time $t$. When passing to the limit $\Delta t\rightarrow0$, this gives the expression
\[ df_t(y)-\frac{\partial f_t(y)}{\partial y}\,dt.  \]
The wealth dynamics generated by a single futures contract investment can now be written as
\begin{equation}\label{Wealth Evolution One Futures Contract}
dX_t=\left(rX_t-\frac{\partial f_t(y)}{\partial y}\right)dt+df_t(y), \ X_0=0.
\end{equation}
Here, $r>0$ denotes the risk free rate of return. It is worth pointing out that the wealth $X$ can become negative.

Consider now a general portfolio of futures contracts. A futures portfolio consists of future obligations to purchase or sell electricity for the futures price. As a starting point, we consider the portfolio strategies to be measure-valued processes $t\mapsto\Gamma(t,\cdot)$. Here, the measure $\Gamma(t,\cdot)$ gives the portfolio weights for contracts with times to maturity $y\geq0$. In other words, it gives the net amount of obligations undertaken by the investor for any given time to maturity. Given a portfolio of futures contracts and the fact that it is costless to take positions in futures, the instantaneous payoff from holding this portfolio is determined solely by the fluctuations of the associated futures prices. This results into the wealth dynamics generated by an arbitrary portfolio strategy $\Gamma$ which can be written as
\begin{equation}\label{Wealth Evolution General}
dX^\Gamma_t=\left(rX^\Gamma_t-\int_0^\infty \frac{\partial f_t(y)}{\partial y}\Gamma(t,dy)\right)dt+\int_0^\infty df_t(y)\Gamma(t,dy), \ X^\Gamma_0=0.
\end{equation}
To make the expression \eqref{Wealth Evolution General} technically tangible, we assume that portfolios $\Gamma:t\mapsto\Gamma(t,\cdot)$ are $\mathbf{F}$-progressively measurable processes taking values in the dual $H^*_\alpha$. Thus we can write the expression \eqref{Wealth Evolution General} in the more compact form
\begin{equation}\label{Wealth Evolution General 002}
dX^\Gamma_t=\left(rX^\Gamma_t-\left\langle \Gamma(t),\frac{\partial f_t(\cdot)}{\partial y}\right\rangle\right)dt+\left\langle\Gamma(t), df_t(\cdot)\right\rangle, \ X^\Gamma_0=0,
\end{equation}
where $\langle\cdot,\cdot\rangle$ is the bilinear pairing on $H_\alpha$. The wealth $X^\Gamma$ can be written as
\begin{equation}\label{Wealth General}
\begin{split}
X_t &=\int_0^t\left(rX^\Gamma_s-\left\langle \Gamma(s),\frac{\partial f_s(\cdot)}{\partial y}\right\rangle\right)ds+\int_0^t\left\langle\Gamma(s), df_s(\cdot)\right\rangle \\
    &=\int_0^t\left(rX^\Gamma_s+\left\langle \Gamma(s),\Sigma(f_t(\cdot))(\psi))\right\rangle\right)ds+\int_0^t\left\langle\Gamma(s), \Sigma(f_s(\cdot)) d\hat{\mathcal{W}}_s \right\rangle,
\end{split}
\end{equation}
for all $t\in[0,T]$. For the right hand side of \eqref{Wealth General} to make sense, we additionally assume that the portfolio process $\Gamma$ satisfies the $L_2(\mathbf{P})$-condition
\begin{equation}\label{Admissibility L_2 condition}
\mathbf{E}\left[ \left(\int_0^T \left\langle\Gamma(t),\Sigma(f_t(\cdot))(\psi) \right\rangle dt \right)^2+\int_0^T\sum_{i\in\mathcal{I}}\sqrt{\lambda_i}\left\langle\,\Gamma(t),\sigma_i(f_t(\cdot))\,\right\rangle^2 dt \right]<\infty.
\end{equation}
Portfolio processes $\Gamma$ satisfying the condition \eqref{Admissibility L_2 condition} will be called \emph{admissible}. Given the time $t$ and the value of the wealth process $X^\Gamma_t=x$, we denote the set of admissible portfolios as $\mathcal{A}(t,x)$.

The portfolio strategies we are considering are not self-financing in the classical meaning of the word. Indeed, since it costs nothing to enter a futures position, the investor needs no initial wealth to set up a portfolio. However, the investor needs credit since the wealth (the amount of money in the external bank account, that is) can become negative. In this case, the investor needs to borrow money to balance the collateral, i.e. to infuse cash into the system. In this sense, the portfolio optimization problem considered here resembles more betting than traditional investment problem. The admissible portfolio strategies $\Gamma$ describe the investor's allocation of wealth over the {\it whole} futures curve. This will include so-called roll-over strategies, where one invests in single a futures contract with a given time to maturity $y$. By picking more times to maturity $y_1,\ldots,y_n$, one can create roll-over portfolios of futures contracts.

\subsection{The Optimization Problem} The objective of the investor is to determine an admissible investment policy which maximizes the expected utility at a terminal time $T$. To study this problem, we first propose an appropriate set of utility functions. We call a function $u:\mathbf{R}\rightarrow[-\infty,\infty)$ a \emph{utility function} if
\begin{itemize}
\item[(1)] it is concave, nondecreasing and upper-semicontinuous,
\item[(2)] the half line $\dom(u):=\{x\in\mathbf{R}\,|\,u(x)>-\infty \}$ is a nonempty subset of $[0,\infty)$,
\item[(3)] $u'$ is continuous, positive and strictly decreasing on the interior of $\dom(u)$ and $\lim_{x\rightarrow\infty}u'(x)=0$,
\end{itemize}
see \cite{Karatzas Shreve}, p. 94. Given a utility function $u$, the portfolio optimization problem reads as
\begin{equation}\label{Merton Problem}
V(t,x)=\sup_{\Gamma\in\mathcal{A}(t,x)}\mathbf{E}\left[u(X^\Gamma_T))\,|\,X^\Gamma_t=x\right], \ V(T,x)=u(x),
\end{equation}
with $0\leq t\leq T$ and $x\geq 0$.

\section{Finite-dimensional Realizations}

\subsection{Invariant Foliation: A Characterization} In this section we assume that the real separable Hilbert space $U$ is truncated to be finite-, say, $n$-dimensional. We can identify the truncated $U$ with the euclidian space $\mathbf{R}^n$. From a practical point of view, this corresponds, for example, to the case where the eigenvalues $\lambda_i$ are very small for all $i>n$. Thus the associated scalar processes $\mathcal{W}^i$, $i>n$ contribute very little to the overall fluctuations of the process $\mathcal{W}$ and, consequently, the processes $\mathcal{W}_i$, $i\leq n$, can be identified as the \emph{principal components} driving the futures curve $f_t$, see, e.g., \cite{Carmona Tehranchi}. We denote the truncated (i.e., $\mathbf{R}^n$-valued) Wiener process as $W$. We remark that the market price of risk $\psi$ degenerates now into a constant vector $(\psi_1,\dots,\psi_n)\in\mathbf{R}^n$. In \cite{GS} a two-factor model for the oil futures price dynamics is proposed, which corresponds in our context to assuming that the futures curve is driven by a two-dimensional Wiener process. On the other hand, in electricity there exist empirical and theoretical evidence for a high degree of idiosyncratic risk, which means that a high dimensional Brownian motion is required for the futures curve dynamics (see \cite{KO,BCK}).

The first objective of this section is to pin down conditions on the volatility structure $\Sigma$ under which the price dynamics given as the solution of the SPDE \eqref{P-forward} admit a finite-dimensional realization. The existence of a finite-dimensional realization means, roughly speaking, that for a given initial curve $h_0$, there exists a nice family of manifolds in $H_\alpha$ such that the initial curve is on the initial manifold and that the evolution of the curve in confined to the family of manifolds. This problem setting has been studied extensively during the last decade or so, references include \cite{Bj�rk Gombani, Bj�rk Svensson, Bj�rk Landen, Bj�rk Blix Landen, Filipovic, Filipovic Teichmann 1, Tappe}. In our analysis we follow the approach of \cite{Tappe}. We proceed by making the following definitions.

\begin{definition}
Let $\mathcal{V}$ be a linear $d$-dimensional subspace of $H_\alpha$. Then
\begin{itemize}
\item[(A)] a family $\mathbf{M}:=(\mathcal{M}_t)_{t\geq0}$ of linear submanifolds of $H_\alpha$ is called a \emph{foliation generated by $\mathcal{V}$}, if the there exist a continuously differentiable $\theta:\mathbf{R}_+\rightarrow H_\alpha$ such that $\mathcal{M}_t=\theta(t)+\mathcal{V}$, for all $t\geq0$. Here, $\theta$ is called the \emph{parametrization} of $\mathbf{M}$. Analogously, the \emph{tangent space} is defined as $T\mathcal{M}_t:=\theta'(t)+\mathcal{V}$ for all $t\geq0$,

\item[(B)] a foliation $\mathbf{M}$ is called \emph{invariant} for the SPDE \eqref{P-forward} if for all $t_0\in\mathbf{R}_+$ and $h_0\in\mathcal{M}_{t_0}$ we have $\mathbf{P} (f_t\in\mathcal{M}_{t_0+t})=1$, $t\geq0$, where $f_t$ is the (weak) solution for \eqref{P-forward} with $f_0=h_0$.
\end{itemize}
\end{definition}

\begin{definition}
Let $\mathcal{V}$ be a linear $d$-dimensional subspace of $H_\alpha$. Then the equation \eqref{P-forward} is said to have a \emph{finite dimensional realization generated by $\mathcal{V}$} if for all $h\in\mathcal{D}(A)$ there exist a foliation $\mathbf{M}^h=(\mathcal{M}^h_t)_{t\geq0}$ generated by $\mathcal{V}$ with $h\in \mathcal{M}^h_0$ which is invariant for \eqref{P-forward}.
\end{definition}

As we see from these definitions, the invariant foliation is the basic building block of a finite-dimensional realization. The next proposition gives necessary and sufficient conditions for invariance of a given foliation for the SPDE \eqref{P-forward}. This proposition is analogous to Theorem 5.3 in \cite{Tappe}, where a similar characterization is proved when the underlying price dynamics are driven by a one-dimensional Wiener process. In our result, we consider the case where the driver is a multi-dimensional Wiener process. In addition, the drift term in the futures price dynamics \eqref{P-forward} cannot be handled immediately using existing results. Indeed, Theorem 5.3 in \cite{Tappe} is concerned with the HJM approach of term structure modeling of {\it forward rates} which leads into the well-known HJM drift condition for the underlying dynamics, whereas we consider a different drift condition stemming from the fact that we analyze futures prices. For convenience, we recall the definition of $\nu$:
\[ \nu(h)=A h+\Sigma(h)(\psi),  \]
where $h\in\mathcal{D}(A)$. Furthermore, we make the following assumption on the volatility structure. These assumptions allow us to use Theorem 2.11 in \cite{Tappe} (We point out that Theorem 2.11 in \cite{Tappe} is proved for a one-dimensional driving Brownian motion but we observe that the result holds also for a multi-dimensional Brownian driver by simply plugging it into the proof).

\begin{assumption}
We assume that the components $\sigma_j$ are continuously differentiable with $\sigma_j(H_\alpha)\in\{ h\in H_{\alpha'} \ : \ h(\infty)=0 \}$, for some $\alpha'>\alpha$, and that there exist $L,M>0$ such that
\begin{displaymath}
\begin{split}
\|\sigma_j(h_1)-\sigma_j(h_2) \|_\alpha&\leq L\|h_1-h_2 \|_\alpha, \ h_1, h_2 \in H_\alpha, \\
\|\sigma_j(h) \|_\alpha&\leq M , \ h\in H_\alpha,
\end{split}
\end{displaymath}
for all $j=1,\dots,n$.
\end{assumption}

\begin{proposition}\label{Invariance characterization proposition}
Let $\mathbf{M}$ be a foliation generated by the $d$-dimensional linear space $\mathcal{V}\subset H_\alpha$ spanned by elements $\{v_i\}_{i=1}^d$. The following statements are equivalent:
\begin{itemize}

\item[(A)] The foliation $\mathbf{M}$ is invariant for the equation \eqref{P-forward}
\item[(B)] We have

\begin{equation}\label{Invariant foliation proposition: condition 1}
\left\{\begin{split}
&\theta(t)\in\mathcal{D}(A), \text{ for all } t\geq0,\\
&v_i\in\mathcal{D}(A), \text{ for all } i=1,\dots,d,
\end{split}\right.
\end{equation}
and there exist functions $\beta\in C^{0,1}(\mathbf{R}_+\times\mathbf{R}^d;\mathbf{R}^d)$ and $\kappa \in C^{0,1}(\mathbf{R}_+\times\mathbf{R}^d;\mathbf{R}^d\times \mathbf{R}^n)$such that
\begin{equation}\label{Invariant foliation proposition: condition 2}
\left\{
\begin{split}
&\nu\left(\theta(t)\right)=\theta'(t)+\sum_{i=1}^d \beta_i(t,0) v_i \in T\mathcal{M}_t, \\
&\sigma_j\left(\theta(t)+\sum_{i=1}^d z_i v_i \right)=\sum_{i=1}^d v_i \kappa_{ij}(t,z)\in\mathcal{V} , \ j=1,\dots,n, \end{split}
\right.
\end{equation}
for all $(t,z)\in\mathbf{R}_+\times\mathbf{R}^d$, and the elements $v_k$ satisfy the ordinary differential equations
\begin{equation}\label{ODEs}
\frac{d}{dy}v_k - \sum_{i=1}^d v_i\frac{\partial}{\partial z_k}\left(\beta_i(t,z)-\sum_{j=1}^n \kappa_{ij}(t,z)\psi_j\right) = 0,
\end{equation}
for all $k=1,\dots,d$.
\end{itemize}
\end{proposition}

\begin{proof}
\emph{Necessity}: Fix $t\geq0$ and assume that the foliation $\mathbf{M}$ is invariant for the equation \eqref{P-forward}. Then we know from \cite{Tappe}, Thrm. 2.11, that $\mathcal{M}_t$ and, consequently, $\theta(t)$ and $v_i$ are in $\mathcal{D}(A)$ for all $t\geq0$ and $i=1,\dots,d$. Let $h\in\mathcal{M}_t$ and write $h=\theta(t)+\sum_{i=1}^d z^h_i v_i$ where $z^h\in\mathbf{R}^d$. Again, we know from \cite{Tappe}, Thrm. 2.11, that $\nu(h)\in T\mathcal{M}_t$ and $\sigma_j(h)\in\mathcal{V}$ for all $j=1,\dots,n$. This yields
\begin{equation}
\sigma_j(h)=\sum_{i=1}^d z^{\sigma_j(h)}_i v_i,
\end{equation}
where $z^{\sigma_j(h)}\in\mathbf{R}^d$ for each $j=1,\dots,n$. Define the linear isomorphism $\mathcal{I}:\mathbf{R}^d\rightarrow \mathcal{V}$ as $\mathcal{I}(y)=\sum_{i=1}^d z_iv_i$. Then we can write
\begin{equation}\label{kappa definition in the proof}
\kappa_{\cdot j}(t,z^h) := z^{\sigma_j(h)} = \mathcal{I}^{-1}\left(\sigma_j\left(\theta(t)+\sum_{i=1}^d z^h_i v_i\right) \right).
\end{equation}
for all $j=1,\dots,n$. The claimed regularity properties of $\kappa$ follow from the assumptions on $\sigma$.

To prove the claim on $\nu$, we find similarly that
\begin{equation}\label{Ah expression}
\nu(h)=\frac{d}{dy}\theta(t)+\sum_{i=1}^d z^h_i \frac{d}{dy} v_i + \sum_{i=1}^d \sum_{j=1}^n v_i z_i^{\sigma_j(h)}\psi_j = \theta'(t) + \sum_{i=1}^d z^{\nu(h)}_i v_i
\end{equation}
where $z^{\nu(h)}\in\mathbf{R}^d$. Then we find using \eqref{kappa definition in the proof} that
\begin{equation}\label{beta definition}
\beta(t,z^h):=z^{\nu(h)} = \mathcal{I}^{-1}\left(\frac{d}{dy}\theta(t)+\sum_{i=1}^d \left(z^h_i \frac{d}{dy} v_i+ \sum_{j=1}^n v_i\kappa_{ij}(t,z^h)\psi_j \right)-\theta'(t) \right).
\end{equation}
Again, the claimed regularity properties of $\nu$ follow from the assumptions on $\sigma$.

Finally, the differential equations \eqref{ODEs} follow by differentiating the expression \eqref{Ah expression} with respect to $z^h_k$ 
for all $k=1,\dots,d$.

\emph{Sufficiency}: Assume now that the conditions \eqref{Invariant foliation proposition: condition 1} -- \eqref{ODEs} hold and let $h\in\mathcal{M}_t$. Using \cite{Tappe}, Thrm. 2.11, it suffices to show that
\begin{equation}\label{Necessary Claim on nu}
\nu\left( \theta(t) + \sum_{i=1}^d z^h_i v_i \right) = \theta'(t) + \sum_{i=1}^d \beta_i(t,z^h) v_i
\end{equation}
to prove the claim. To this end, we observe first that $\sigma_j(\theta(t))=v_i\kappa_{ij}(t,0)$ for all $j=1,\dots,n$. Thus we can write
\begin{equation}
\begin{split}
\nu&\left( \theta(t) + \sum_{i=1}^d z^h_i v_i \right) \\
 &= \nu(\theta(t)) + \sum_{i=1}^d z^h_i \frac{d}{dy} v_i + \sum_{j=1}^n\sum_{i=1}^d v_i(\kappa_{ij}(t,z^h)-\kappa_{ij}(t,0))\psi_j \nonumber\\
&= \theta'(t)+ \sum_{i=1}^d \left(v_i\left(\beta_i(t,0)+\sum_{j=1}^n(\kappa_{ij}(t,z^h)-\kappa_{ij}(t,0))\psi_j \right)+z^h_i \frac{d}{dy}v_i \right).\nonumber
\end{split}
\end{equation}
By using the differential equations \eqref{ODEs} and changing the order of summation, we find that
\begin{align*}
\sum_{i=1}^d& \left(v_i\left(\beta_i(t,0)+\sum_{j=1}^n(\kappa_{ij}(t,z^h)-\kappa_{ij}(t,0))\psi_j \right)+z^h_i \frac{d}{dy}v_i \right) \\
&\qquad=\sum_{i,k=1}^d \left(v_i\left(\beta_i(t,0)+\sum_{j=1}^n(\kappa_{ij}(t,z^h)-\kappa_{ij}(t,0))\psi_j \right) \right.\\
&\qquad\qquad\left.+ z^h_iv_k\frac{\partial}{\partial z_i}\left(\beta_k(t,z^h)-\sum_{j=1}^n\kappa_{kj}(t,z^h)\psi_j\right)\right) \\
&\qquad=\sum_{i,k=1}^d \left(\left(\beta_k(t,0)-\sum_{j=1}^n \kappa_{kj}(t,0)\psi_j \right) \right. \\
&\qquad\qquad\left.+ z^h_i\frac{\partial}{\partial z_i}\left(\beta_k(t,z^h)-\sum_{j=1}^n \kappa_{kj}(t,z^h)\psi_j \right)+\sum_{j=1}^n \kappa_{kj}(t,z^h)\psi_j \right)v_k.
\end{align*}
To proceed, we readily verify that
\begin{equation}\label{beta computation}
\begin{split}
\beta(t,z)&-\sum_{j=1}^n \kappa_{\cdot j} (t,z)\psi_j  \\
&= \beta(t,z)-\left(\mathcal{I}^{-1}\circ\mathcal{I}\right)\left(\sum_{j=1}^n \kappa_{\cdot j} (t,z)\psi_j\right) \\
&= \mathcal{I}^{-1}\left( \frac{d}{dy}\theta(t) - \sum_{i=1}^d z^h_i \frac{d}{dy}v_i - \theta'(t) \right)
\end{split}
\end{equation}
On the other hand, since $\mathcal{I}$ is a linear isomorphism, the inverse $\mathcal{I}^{-1}:\mathcal{V}\rightarrow\mathbf{R}^d$ is also linear. Thus we observe from the expression \eqref{beta computation} that the partial derivatives $\frac{\partial}{dz_i}(\beta(t,z) - \sum_{j=1}^n\kappa_{\cdot j}(t,z) \psi_j)$ are independent of $z$. Consequently, we find that
\[ \beta(t,0)-\sum_{j=1}^n \kappa_{\cdot j} (t,0)\psi_j + \sum_{i=1}^d z_i\frac{\partial}{\partial z_i}\left(\beta(t,z)-\sum_{j=1}^n \kappa_{\cdot j}(t,z)\psi_j \right) \]
is the Maclaurin series of the function $\beta-\sum_{j=1}^n\kappa_{\cdot j}\psi_j$. This proves the identity \eqref{Necessary Claim on nu}.
\end{proof}

Proposition \ref{Invariance characterization proposition} gives a convenient characterization of invariance of a given foliation. First, it gives a set of ordinary differential equations \eqref{ODEs} which have the spanning functions $v_i$ as their solutions. The coefficients of these ODEs are characterized in terms of the original volatility structure $\Sigma$ and the market price of risk $\psi$. Furthermore, we observe that the volatility $\Sigma$ must be of a very specific type in order to be associable to an affine realization. Indeed, each of the components $\sigma_j$ must map the manifold $\mathcal{M}$ into the linear space $\mathcal{V}$ for all $t$. This means that the randomness in the price dynamics is confined to the finite-dimensional linear structure $\mathcal{V}$. Put differently, the price dynamics admitting a finite-dimensional realization can diffuse only in finitely many directions in the infinite-dimensional state space $H_\alpha$.

We can identify the coordinate process $Z$ driving the finite-dimensional realization using Proposition \ref{Invariance characterization proposition}. To this end, fix $t_0\in\mathbf{R}_+$ and $h\in\mathcal{M}_{t_0}$. Then we have a unique $z\in\mathbf{R}^d$ such that $h=\theta(t_0)+\sum_{i=1}^d z_i v_i$. Now, let $Z$ be the strong solution of the It\^{o} equation
\[ dZ_t = \beta(t_0+t,Z_t)dt + \kappa(t_0+t,Z_t)dW_t, \ Z_0=z,  \]
where $\beta$ and $\kappa$ are given by Proposition \ref{Invariance characterization proposition}. Then it is straightforward to verify using It\^{o}'s formula that the $H_\alpha$-valued process $f$ defined as
\begin{equation}\label{FDR expression}
f_t:=\theta(t_0+t)+\sum_{i=1}^d Z^i_t v_i, \  f_0=h,
\end{equation}
is the strong solution of SPDE \eqref{P-forward}.

\subsection{Interpretation of the Coordinate Process}\label{Interpretation Section} The coordinate process $Z$ has \emph{a priori} no intrinsic economical interpretation. However, we can equip it with one. To this end, we observe that the futures curve can be decomposed as
\[ f_t=\pi_\mathcal{V} f_t + \pi_{\mathcal{V}^\bot} f_t = \pi_\mathcal{V} f_t + \pi_{\mathcal{V}^\bot} \mathcal{M}_{t+t_0}.  \]
Here, the latter equality (which holds almost surely) follows from the invariance of the foliation $\mathbf{M}$ and the fact that the projection $\pi_{\mathcal{V}^\bot} \mathcal{M}_{t+t_0}$ is a singleton set. We point out that the projection $\pi_{\mathcal{V}^\bot} \mathcal{M}_{t+t_0}$ can be used as a parametrization of the foliation and that it is the unique parametrization which is in $\mathcal{V}^{\bot}$ for all $t\geq0$. According to \eqref{FDR expression}, we can rewrite the futures price as
\begin{equation}
f_t=\pi_{\mathcal{V}^\bot} \mathcal{M}_{t+t_0}+\sum_{i=1}^d Z^i_t v_i.
\end{equation}
Denote a basis of the subspace $\mathcal{V}^\bot$ as $\{ w_j \}_{j=1}^\infty$ and let $\Lambda:H_\alpha\rightarrow \mathbf{R}^d$ be a linear continuous operator such that $\Lambda(\mathcal{V})=\mathbf{R}^d$. Then we can write
\begin{equation}
\begin{split}
\Lambda(f_t)        &=\Lambda(\pi_{\mathcal{V}^\bot} \mathcal{M}_{t+t_0})+\Lambda\left(\sum_{i=1}^d Z^i_t v_i\right)=\sum_{j=1}^\infty c^j_t\Lambda(w_j)+\sum_{i=1}^d Z^i_t \Lambda(v_i)
    \\              &=\sum_{j=1}^\infty c_t^j\sum_{i=1}^d b_{ji}\Lambda(v_i)+\sum_{i=1}^d Z^i_t \Lambda(v_i)=\sum_{i=1}^d\left( Z^i_t+\sum_{j=1}^\infty c_t^jb_{ji} \right)\Lambda(v_i).
\end{split}
\end{equation}
Define the matrix $\Theta$ as $\Theta_{ij}=\Lambda_i(v_j)$, where $i,j=1,\dots,d$. Then we can write
\[ Z_t=\Lambda(f_t)\Theta^{-1}-\sum_{j=1}^\infty c_t^jb_{j\cdot}. \]
In other words, we observe that the state process $Z$ can be associated to the image of $f$ in arbitrary linear operator $\Lambda$ modulo an affine transformation depending of $\Lambda$ and the invariant foliation $\mathbf{M}$. This means, in particular, that the coordinate process $Z$ becomes an observable quantity as a functional transformation of quantities which are observed from the futures price curve.

\begin{example}\label{Benchmark example}
Consider the operator $\Lambda$ given by \emph{benchmark contracts for times to maturity $y_i$}, $i=1,\dots,d$, defined as
\[ \Lambda_i(h)=\delta_{y_i}(h), \ h\in H_\alpha. \]
Here, $\delta$ is the evaluation functional on $H_\alpha$, $\delta_y(h)=h(y)$. Then we can write
\[ Z^i_t=\sum_{j=1}^d f_t(y_j)\Theta_{ji}^{-1}-\sum_{j=1}^\infty c_t^jb_{ji}, \]
for all $i=1,\dots,d$. The dimension $d$ gives the number of benchmark contracts needed to reconstruct the whole futures curve $t\mapsto f_t$ via the identity \eqref{FDR expression}. Indeed, we need $d$ benchmark contracts in order to have a connection between the coordinate process $Z$ and the benchmark contracts $f_\cdot(y_i)$. This connection is given by the invertible $d\times d$-matrix $\Theta$ with coefficients $\Lambda_i(v_j)=v_j(y_i)$, $i,j=1,\ldots,d$.
\end{example}

\subsection{Existence of Finite-dimensional Realization: Sufficient Conditions} In the previous subsection, we proved a characterization of the invariance of a given foliation $\mathbf{M}$ for the SPDE \eqref{P-forward}. In this subsection we use this result to give sufficient conditions for the existence of a finite-dimensional realization. To this end, consider the futures price dynamics given by \eqref{P-forward} when the volatility is of the form
\begin{equation}\label{Sigma Shape}
\sigma_j(h)=\sum_{i=1}^p v_i \Phi_{ij}(h), \ h \in H_\alpha,
\end{equation}
for all $j=1,\dots,n$. Here, the functions $v_i$, $i=1,\dots,p$, are linearly independent and each $\Phi_{ij}$ maps $H_\alpha$ into $\mathbf{R}$. We assume that the functionals $\Phi_{ij}$ are twice continuously differentiable and that there exist constants $L,M>0$ such that
\begin{displaymath}
\begin{split}
\|\Phi_{ij}(h_1)-\Phi_{ij}(h_2) \|_\alpha&\leq L\|h_1-h_2 \|_\alpha, \ h_1, h_2 \in H_\alpha, \\
\|\Phi_{ij}(h) \|_\alpha&\leq M , \ h\in H_\alpha,
\end{split}
\end{displaymath}
for all $i=1,\dots,d$ and $j=1,\dots,n$. This guarantees that we can use Proposition \ref{Invariance characterization proposition}. The equation \eqref{P-forward} can now be written as
\begin{equation}\label{P-dynamics 002}
df_t(\cdot)= \left(Af_t(\cdot)+\sum_{j=1}^n \sum_{i=1}^p v_i \Phi_{ij}(f_t(\cdot))\psi_j\right)dt+ \sum_{j=1}^n \sum_{i=1}^p v_i \Phi_{ij}(f_t(\cdot))d\hat{W}^j_t, \ f_0=f,
\end{equation}
where $\hat{W}^j$ are scalar $\mathbf{P}$-Brownian motions. We remark that the expression \eqref{Sigma Shape} gives actually a necessary condition for the existence of a finite-dimensional realization, see \cite{Tappe}, Lemma 3.2. Thus this form of volatility is the most general we can consider in this framework. The following proposition, which is analogous to Proposition 6.2 in \cite{Tappe}, gives sufficient conditions for the existence of a finite-dimensional realization.

\begin{proposition}\label{Invariance sufficient conditions proposition}
Assume that the volatility structure of \eqref{P-forward} is given by \eqref{Sigma Shape} and that the functions $v_i$, $i=1,\dots,p$ are quasi-exponential. Then the equation \eqref{P-forward} admits a finite-dimensional realization.
\end{proposition}

\begin{proof}
Since the elements $v_i$, $i=1,\dots,p$ are quasi-exponential, the linear space
\begin{displaymath}
Y:=\bigoplus_{i=1}^p \Span \left\{ \frac{d^n}{dy^n} v_i \ : \ n \geq 0 \right\} \subset \mathcal{D}(A)
\end{displaymath}
is, by definition, finite-dimensional. Furthermore, we observe that
\begin{equation}\label{derivative condition 001}
\frac{d}{dy}v \in Y, \ v\in Y.
\end{equation}
Set $d:=\dim Y$ and choose $v_{p+1},\dots,v_d\in Y$ such that $\{v_i,\dots,v_d\}$ is a basis of $Y$.

Fix an arbitrary $h_0\in\mathcal{D}(A)$. First, define the function $\theta(t):\mathbf{R}_+\rightarrow H_\alpha$ as
\begin{equation}
\theta(t) := S_th_0 \in \mathcal{D}(A).
\end{equation}
and the function $\kappa\in C^{0,1}(\mathbf{R}_+\times\mathbf{R}^d;\mathbf{R}^d\times\mathbf{R}^n)$ as
\begin{equation}\label{sufficient proposition kappa definition}
\kappa_{ij}(t,z)=
\begin{cases}
\Phi_{ij}\left(\theta(t)+\sum_{k=1}^d z_k v_k\right), & i=1,\dots,p \\
0, & i=p+1,\dots,d.
\end{cases}
\end{equation}
for all $j=1,\dots,n$. We observe from \eqref{sufficient proposition kappa definition} that the latter condition in \eqref{Invariant foliation proposition: condition 2} is satisfied. Furthermore, define the function $\beta\in C^{0,1}(\mathbf{R}_+\times\mathbf{R}^d;\mathbf{R}^d)$ as
\begin{equation}\label{sufficient proposition beta definition}
\beta_i(t,z)=\sum_{j=1}^n\kappa_{ij}(t,z)\psi_j+\sum_{k=1}^d a_{ik}z_k,
\end{equation}
where $a_{ij}$ are chosen, due to \eqref{derivative condition 001}, such that
\begin{equation}\label{Matrix A defining expression}
\frac{d}{dy}v_i = \sum_{j=1}^d  v_j a_{ji} ,
\end{equation}
for all $i=1,\dots,d$. With this specification, we verify by differentiating expression \eqref{sufficient proposition beta definition} with respect to $z_k$ that the elements $v_i$ satisfy the differential equations \eqref{ODEs}. Finally, the definition of $\beta$ coupled with the fact that $\frac{d}{dy}\theta(t)=\theta'(t)={h_0}'(y+t)$ implies that the former condition in \eqref{Invariant foliation proposition: condition 1} is also satisfied. 
\end{proof}

Proposition \ref{Invariance sufficient conditions proposition} and its proof give us not only sufficient conditions for the existence of a finite-dimensional realization but also a recipe for its construction. To illustrate this, we consider the example from Section 6 in \cite{Bj�rk Gombani}.
\begin{example}
Assume the driving Brownian motion in \eqref{P-dynamics 002} is one-dimensional, the volatility functional $\Phi\equiv 1$ and the quasi-exponential function $v_1(x)=x e^{-ax}$ for $a>0$, i.e., that $\sigma=\sigma_1=v_1$. In addition, let for simplicity $\psi=0$. We observe that the space $\mathcal{V}=\Span \{ v_1, v_2 \}$, where $v_2(y)=e^{-ay}$, so the dimension of the affine realization is $2$. It is a matter of differentiation to show (see \eqref{Matrix A defining expression})
\begin{displaymath}
\begin{pmatrix}
v_1' & v_2'
\end{pmatrix}
=
\begin{pmatrix}
v_1 & v_2
\end{pmatrix}
\begin{pmatrix}
-a & 0 \\ 1 & -a
\end{pmatrix}:=v \mathbf{A}.
\end{displaymath}
Thus the state variable $Z$ is given, due to \eqref{sufficient proposition kappa definition} and \eqref{sufficient proposition beta definition}, as the strong solution of the two-dimensional It\^{o} equation
\[ dZ_t = \mathbf{A}Z_tdt+ \mathbf{B}dW_t,  \]
where $\mathbf{B}=(1,0)^\top$ and $W$ is a $\mathbf{Q}$-Brownian motion. Now, given the initial $h_0\in\mathcal{D}(A)$, the affine realization of futures price reads as
\[ f_t(\cdot)=h_0(\cdot+t_0+t)+Z^1_tv_1(\cdot)+Z^2_tv_2(\cdot).  \]
It is worth pointing out that the process $Z$ becomes an Ornstein-Uhlenbeck process with solution given by
$$
Z_t=\exp(\mathbf{A} t)Z_0+\int_0^t\exp(\mathbf{A}(t-s))\mathbf{B}\,dW_s\,.
$$
As the eigenvalue of the matrix $\mathbf{A}$ is $-a$, that is, negative, $Z_t$ becomes stationary. Many popular spot and futures
price dynamics of commodities and energy have stationarity as a crucial property, see \cite{BSBK} for examples in energy. The model in
the current example gives rise to a hump-shaped futures price curve, which is relevant in oil markets, say (see \cite{G}).
\end{example}

\section{Portfolio Optimization Revisited: Finite-dimensional State Variable} In this section we recast the optimization problem \eqref{Merton Problem} into a finite-dimensional setting using the identity \eqref{FDR expression}. Assume that the initial futures price curve can be expressed as $f_0 = S_{t_0} h_0 + \sum_{i=1}^d z_i v_i$ where $h_0\in\mathcal{D}(A)$ and $t_0\in\mathbf{R}_+$. Then the affine futures price curve dynamics is given by
\[ f_t = S_tS_{t_0}h_0 + \sum_{i=1}^d Z_t^i v_i.  \]
Recall that the wealth $X$ is given by the expression \eqref{Wealth General}. To rewrite this for the affine realization, we first observe that
\[ df_t = S_tS_{t_0}\frac{\partial h_0}{\partial y}dt + \sum_{i=1}^d dZ_t^i v_i. \]
Then we can rewrite the wealth as
\begin{equation}\label{Wealth affine dynamics 001}
X^\Gamma_t=\int_0^t\left(rX^\Gamma_s-\left\langle\Gamma(s),\frac{\partial f_s}{\partial y} \right\rangle\right)ds + \int_0^t \left\langle\Gamma(s),\sum_{i=1}^d dZ^i_sv_i  \right\rangle + \int_0^t \left\langle\Gamma(s),S_sS_{t_0}\frac{\partial h_0}{\partial y}\right\rangle ds.
\end{equation}
It is a matter of differentiation to show that
\begin{equation*}
S_sS_{t_0}\frac{\partial h_0}{\partial y}-\frac{\partial f_s}{\partial y} = S_sS_{t_0} \frac{\partial h_0}{\partial y}  - S_sS_{t_0}\frac{\partial h_0}{\partial y} - \sum_{i=1}^d Z^i_s\frac{\partial v_i}{\partial y}= - \sum_{i=1}^d Z^i_s\frac{\partial v_i}{\partial y}.
\end{equation*}
By coupling this with the expression \eqref{Matrix A defining expression}, we observe that the equation \eqref{Wealth affine dynamics 001} can be written as
\begin{equation}\label{Wealth affine dynamics 002}
\begin{split}
X^\Gamma_t &=\int_0^t r X^\Gamma_s ds + \int_0^t \left\langle\Gamma(s), \sum_{i=1}^d dZ^i_sv_i- \sum_{i=1}^d Z^i_s\frac{\partial v_i}{\partial y}ds\right\rangle \\
    &=\int_0^t r X^\Gamma_s ds + \int_0^t \left\langle\Gamma(s), \sum_{i=1}^d dZ^i_sv_i- \sum_{i=1}^d Z^i_s \sum_{j=1}^d v_j a_{ji}ds\right\rangle\\
    &=\int_0^t r X^\Gamma_s ds+\int_0^t \left\langle\Gamma(s), \sum_{i=1}^d v_i\left(dZ^i_s-\sum_{j=1}^d a_{ij}Z^j_sds\right)\right\rangle\\
    &=\int_0^t r X^\Gamma_s ds+\int_0^t \sum_{i=1}^d \gamma^i_t\left(dZ^i_s-\sum_{j=1}^d a_{ij}Z^j_sds\right),
\end{split}
\end{equation}
where the real valued processes $\gamma^i:=\langle \Gamma,v_i \rangle$ satisfy the condition \eqref{Admissibility L_2 condition}. Finally, by using identity \eqref{sufficient proposition beta definition}, we find
\begin{equation}\label{Wealth affine dynamics 004}
\begin{split}
X^\Gamma_t= X^\gamma_t &=\int_0^t r X^\gamma_s ds + \int_0^t \sum_{i=1}^d \gamma^i_s\sum_{j=1}^n\kappa_{ij}(s,Z_s)\left(d\hat{W}^j_s+\psi_jds\right) \\
    &=\int_0^t r X^\gamma_s ds + \int_0^t \gamma_s(dZ_s-\mathbf{A}Z_sds), \\
\end{split}
\end{equation}
where $\hat{W}$ is a $\mathbf{P}$-Wiener process and the matrix $\mathbf{A}$ is given by \eqref{Matrix A defining expression}. By plugging this into the utility maximization problem \eqref{Merton Problem}, we reduce the initial infinite-dimensional control problem into a classical finite-dimensional control problem. We write this problem in the Markovian form
\begin{equation}\label{Finite dimensional control problem}
V(t,x,z)=\sup_{\gamma}\mathbf{E}\left[u(X^\gamma_T) |  X^\gamma_t=x, Z_t=z \right].
\end{equation}
We can use classical control theoretic techniques to study this problem, see, e.g., \cite{Fleming Soner}. To illustrate this, we derive the HJB-equation for this problem. Assuming that the function $V$ is sufficiently smooth, It\^{o}'s formula yields
\begin{equation}
\begin{split}
dV(t,X^\gamma_t,Z_t)=&\left[\vphantom{\frac{1}{2}}V_t(t,X^\gamma_t,Z_t)+ V_x(t,X^\gamma_t,Z_t)\left(rx+\sum_{i=1}^d \gamma_i\sum_{k=1}^n \kappa_{ik}(t,Z_t)\psi_k \right)+\sum_{i=1}^d V_{z_i}(t,X^\gamma_t,Z_t)\beta_i(t,Z_t)\right. \\ &+\left.\frac{1}{2}V_{xx}(t,X^\gamma_t,Z_t)\sum_{i=1}^d\sum_{k=1}^n\gamma_i\kappa_{ij}(t,Z_t)\kappa_{ji}(t,Z_t)\gamma_i
+\frac{1}{2}\sum_{i,j=1}^d V_{z_iz_j}(t,X^\gamma_t,Z_t)\sum_{k=1}^n \kappa_{jk}(t,Z_t)\kappa_{ki}(t,Z_t)\right. \\ &+\left.\sum_{i=1}^d V_{xz_i}(t,X^\gamma_t,Z_t)\sum_{k=1}^n\kappa_{ik}(t,Z_t)\kappa_{ki}(t,Z_t)\gamma_i \right]dt \\ &+ V_x(t,X^\gamma_t,Z_t)\sum_{i=1}^d \gamma_i\sum_{k=1}^n \kappa_{ik}(t,Z_t)d\hat{W}^k_t+\sum_{i=1}^d V_{z_i}(t,X^\gamma_t,Z_t)\sum_{k=1}^n\kappa_{ik}(t,Z_t) d\hat{W}^k_t.
\end{split}
\end{equation}
By a standard martingale argument from stochastic control, see, e.g., \cite{Fleming Soner}, this yields the HJB-equation
\begin{equation}
\begin{split}
\sup_{\gamma}&\left\{\vphantom{\frac{1}{2}}V_t(t,x,z)+ V_x(t,x,z)\left(rx+\sum_{i=1}^d \gamma_i\sum_{k=1}^n \kappa_{ik}(t,z)\psi_k \right)+\sum_{i=1}^d V_{z_i}(t,x,z)\beta_i(t,z)\right. \\ &+\left.\frac{1}{2}V_{xx}(t,x,z)\sum_{i=1}^d\sum_{k=1}^n\gamma_i\kappa_{ij}(t,z)\kappa_{ji}(t,z)\gamma_i
+\frac{1}{2}\sum_{i,j=1}^d V_{z_iz_j}(t,x,z)\sum_{k=1}^n \kappa_{jk}(t,z)\kappa_{ki}(t,z)\right. \\ &+\left.\sum_{i=1}^d V_{xz_i}(t,x,z)\sum_{k=1}^n\kappa_{ik}(t,z)\kappa_{ki}(t,z)\gamma_i\right\}=0,
\end{split}
\end{equation}
where $V(T,x,z)=u(x)$ for all $x,z\in\mathbf{R}$. As a consequence, we can formulate the following verification result

\begin{proposition}\label{verification}
Assume that we have
\begin{itemize}
\item[(1)] a smooth function $\hat{V}:[0,T]\times\mathbf{R}^d\times \mathbf{R}\rightarrow \mathbf{R}$ with $\hat{V}(T,z,x)=u(x)$ such that
\begin{equation*}
\begin{split}
\hat{V}_t(t,x,z)&+ \hat{V}_x(t,x,z)\left(rx+\sum_{i=1}^d \gamma_i\sum_{k=1}^n \kappa_{ik}(t,z)\psi_k \right)+\sum_{i=1}^d V_{z_i}(t,x,z)\beta_i(t,z) \\ &+\frac{1}{2}\hat{V}_{xx}(t,x,z)\sum_{i=1}^d\sum_{k=1}^n\gamma_i\kappa_{ij}(t,z)\kappa_{ji}(t,z)\gamma_i
+\frac{1}{2}\sum_{i,j=1}^d \hat{V}_{z_iz_j}(t,x,z)\sum_{k=1}^n \kappa_{jk}(t,z)\kappa_{ki}(t,z) \\ &+\sum_{i=1}^d \hat{V}_{xz_i}(t,x,z)\sum_{k=1}^n\kappa_{ik}(t,z)\kappa_{ki}(t,z)\gamma_i\leq 0,
\end{split}
\end{equation*}
for all $t,\gamma,x,z$,
\item[(2)] a function $\gamma^*:[0,T]\times \mathbf{R}\times \mathbf{R}^d\rightarrow \mathbf{R}^d$ such that
\begin{equation*}
\begin{split}
\hat{V}_t(t,x,z)&+ \hat{V}_x(t,x,z)\left(rx+\sum_{i=1}^d \gamma^*_i(t,x,z)\sum_{k=1}^n \kappa_{ik}(t,z)\psi_k \right)+\sum_{i=1}^d V_{z_i}(t,x,z)\beta_i(t,z) \\ &+\frac{1}{2}\hat{V}_{xx}(t,x,z)\sum_{i=1}^d\sum_{k=1}^n\gamma^*_i(t,x,z)\kappa_{ij}(t,z)\kappa_{ji}(t,z)\gamma^*_i(t,x,z)
\\&+\frac{1}{2}\sum_{i,j=1}^d \hat{V}_{z_iz_j}(t,x,z)\sum_{k=1}^n \kappa_{jk}(t,z)\kappa_{ki}(t,z) \\ &+\sum_{i=1}^d \hat{V}_{xz_i}(t,x,z)\sum_{k=1}^n\kappa_{ik}(t,z)\kappa_{ki}(t,z)\gamma^*_i(t,x,z)= 0,
\end{split}
\end{equation*}
for all $t,z,x$.
\item[(3)] a wealth process $X^*$ corresponding to an admissible control $\hat{\gamma}$ such that $\hat{\gamma}_t=\gamma^*(t,X^*_t,z)$ for all $t\geq0$ and $z\in\mathbf{R}^d$.
\end{itemize}
Then the function $\hat{V}$ solves the control problem \eqref{Finite dimensional control problem} (i.e. $\hat{V}=V$) and an optimal control is given by $\hat{\gamma}$.
\end{proposition}

\begin{proof}
Let $\gamma$ be an admissible control and $X^\gamma$ the associated wealth process. Then It\^{o}'s formula yields
\begin{equation}
\begin{split}
d\hat{V}(t,X^\gamma_t,Z_t)=&\left[\vphantom{\frac{1}{2}}\hat{V}_t(t,X^\gamma_t,Z_t)+ \hat{V}_x(t,X^\gamma_t,Z_t)\left(rx+\sum_{i=1}^d \gamma_i(t)\sum_{k=1}^n \kappa_{ik}(t,Z_t)\psi_k \right)+\sum_{i=1}^d \hat{V}_{z_i}(t,X^\gamma_t,Z_t)\beta_i(t,Z_t)\right. \\ &+\left.\frac{1}{2}\hat{V}_{xx}(t,X^\gamma_t,Z_t)\sum_{i=1}^d\sum_{k=1}^n\gamma_i(t)\kappa_{ij}(t,Z_t)\kappa_{ji}(t,Z_t)\gamma_i(t)
+\frac{1}{2}\sum_{i,j=1}^d \hat{V}_{z_iz_j}(t,X^\gamma_t,Z_t)\sum_{k=1}^n \kappa_{jk}(t,Z_t)\kappa_{ki}(t,Z_t)\right. \\ &+\left.\sum_{i=1}^d \hat{V}_{xz_i}(t,X^\gamma_t,Z_t)\sum_{k=1}^n\kappa_{ik}(t,Z_t)\kappa_{ki}(t,Z_t)\gamma_i(t) \right]dt \\ &+ \hat{V}_x(t,X^\gamma_t,Z_t)\sum_{i=1}^d \gamma_i(t)\sum_{k=1}^n \kappa_{ik}(t,Z_t)d\hat{W}^k_t+\sum_{i=1}^d \hat{V}_{z_i}(t,X^\gamma_t,Z_t)\sum_{k=1}^n\kappa_{ik}(t,Z_t) d\hat{W}^k_t.
\end{split}
\end{equation}
Since the term on the right hand side inside the angled brackets is negative, we obtain
\begin{equation*}
\begin{split}
u(X^\gamma_T)=\hat{V}(T,X^\gamma_T,Z_T)\leq \hat{V}(t,X^\gamma_t,Z_t)+ \int_t^T &\left(\hat{V}_x(s,X^\gamma_s,Z_s)\sum_{i=1}^d \gamma_i(s)\sum_{k=1}^n \kappa_{ik}(s,Z_s)d\hat{W}^k_s\right.\\&+\left.\sum_{i=1}^d \hat{V}_{z_i}(s,X^\gamma_s,Z_s)\sum_{k=1}^n\kappa_{ik}(s,Z_s) d\hat{W}^k_s\right).
\end{split}
\end{equation*}
Since $\gamma$ is admissible, the stochastic integral is a true martingale and, consequently, $\mathbf{E}[u(X^\gamma_T)|X^\gamma_t=x,Z_t=z]\leq\hat{V}(t,X^\gamma_t,Z_t)$.
This shows that $\hat{V}$ dominates the solution $V$. To show that $\hat{V}\leq V$ and that $\hat{\gamma}$ gives an optimal control, we use the same argument coupled with the conditions (2) and (3).
\end{proof}

Proposition \ref{verification} gives a set of sufficient conditions for a given smooth function to coincide with the value of the control problem \eqref{Finite dimensional control problem}. Moreover, it gives an identification of an optimal control. Even though the control $\gamma$ does not appear to have an direct economical interpretation, it can be traced back to, for example, to benchmark futures prices following Example \ref{Benchmark example}. In this example, the $d$-dimensional coordinate process $Z$ was identified with an affine transform of $d$ benchmark futures prices. On the other hand, as we observe from the equation \eqref{Wealth affine dynamics 004}, the process $\gamma$ can be understood as "portfolio weights" on $d$-dimensional diffusion dynamics $t\mapsto \int_0^t \kappa(s,Z_s)\left(d\hat{W}_s+\psi ds\right)$. These dynamics can extracted from the evolution of the coordinate process $Z$.

We end our paper with discussing a futures price dynamics typical for energy markets. Suppose that the dynamics is given as
\begin{equation}
\label{GS-futures}
df_t(x)=\left(Af_t(x)+\sigma_1\psi_1+\sigma_2 e^{-a x}\psi_2\right)\,dt+\sigma_1\,d\hat{W}_t^1+\sigma_2 e^{-a x}\,d\hat{W}_t^2,
\end{equation}
for a two-dimensional Brownian motion $\hat{W}$ and $\sigma_1, \sigma_2, a$ positive constants. For simplicity, we also assume the market price of risk $\psi_i, \ i=1,2$ to be constants.  In this model, we identify $\Phi_{12}=\Phi_{21}=0$, $\Phi_{11}=\sigma_1$ and $\Phi_{22}=\sigma_2$. Furthermore, $v_1(x)=1$ and $v_2(x)=\exp(-a x)$, which obviously are quasi-exponential functions spanning a linear space of dimension 2. As $v_1'(x)=0$ and $v_2'(x)=-a v_2(x)$, we find the dynamics of the two-dimensional process $Z$ to be
$$
dZ_t=\left(\Psi+\mathbf{A}Z_t\,\right)dt+\mathbf{B}\,d\hat{W}_t,
$$
with $\Psi=(\sigma_1\psi_1,\sigma_2\psi_2)'$ and
$$
\mathbf{A}=\left[\begin{array}{cc} 0 & 0 \\ 0 & -a\end{array}\right],\qquad
\mathbf{B}=\left[\begin{array}{cc} \sigma_1 & 0 \\ 0 & \sigma_2\end{array}\right].
$$
We conclude by the analysis above on finite dimensional realizations that
\begin{equation}
\label{GS-expl}
f_t(x)=h_0(x+t_0+t)+Z_1^1+e^{-a x}Z_t^2,
\end{equation}
for some $t_0>0$. By splitting into the two components of $Z$, we find that $Z^1$ is a drifted Brownian motion
$$
dZ_t^2=\sigma_1\psi_1\,dt+\sigma_1\,d\hat{W}_t^1,
$$
and $Z^2_t$ is a mean-reverting process,
$$
dZ_t^2=(\sigma_2\psi_2-a Z_t^2)\,dt+\sigma_2\,d\hat{W}_t^2.
$$
Optimizing a portfolio invested in the futures price curve will with this model be equivalent to optimizing a portfolio investment in two "assets" with dynamics $Z$. As $\exp(-ax)$ tends to zero when $x$ tends to infinity, we find that $f_t(x)\sim h_0(x+t_0+t)+Z_t^1$ for large values of $x$. Hence, an investment in $Z^1$ can be viewed as holding a portfolio position in a futures with long time to maturity, that is, a position in a contract in the far end of the futures curve. The "asset" $Z^2$ can then be interpreted as the difference between a futures far out on the curve (being $Z^1$) and one with short time to maturity. Our investment problem will therefore be to select optimally a portfolio of contracts in the short and long end of the curve.

The model in \eqref{GS-futures}, or \eqref{GS-expl}, can be viewed as the implied futures price dynamics from a two-factor spot model. In fact, following \cite{GS}, we can assume that the spot price of some commodity is given by
$$
S_t=\Lambda(t)+Z_t^1+Z_t^2,
$$
where $\Lambda(t)$ is some deterministic seasonality function. This spot price model will be an arithmetic analogue of the dynamics proposed for oil by \cite{GS}. Here, $Z^2$ is interpreted as the short term variations of the oil spot price, while $Z^1$, the non-stationary part, is the long-term trends in oil prices including inflation and extinction of reserves. This corresponds to the view of investing in the long and short end of the futures curve.

\subsection*{Acknowledgements}
Financial support from the project "Energy markets: modelling, optimization and simulation (EMMOS)", funded by the Norwegian Research Council under grant 205328 is gratefully acknowledged.


\begin{thebibliography}{AO}



\bibitem{BCK}
Benth, F. E., Cartea, A. and Kiesel, R. (2008). Pricing forward contracts in power markets by the certainty equivalence
principle: explaining the sign of the market risk premium,  {\it Journal of Banking and Finance}, 32/10, 2006 -- 2021.

\bibitem{BSBK}
Benth, F. E., \v{S}altyt\.{e} Benth, J. and Koekebakker, S. (2008). {\it Stochastic Modelling of Electricity and Related Markets}, World
Scientific.


\bibitem{Bj�rk Gombani}
Bj\"{o}rk, T. and Gombani, A. (1999). Minimal realizations of interest rate models, \emph{Finance and Stochastics}, 3, 413 -- 432

\bibitem{Bj�rk Svensson}
Bj\"{o}rk, T. and Svensson, L. (2001). On the existence of finite-dimensional realizations for nonlinear forward rate models, \emph{Mathematical Finance}, 11/2, 205 -- 243

\bibitem{Bj�rk Landen}
Bj\"{o}rk, T. and Land\'{e}n, C. (2002). On the construction of finite dimensional realizations for nonlinear forward rate models, \emph{Finance and Stochastics}, 6, 303 -- 331

\bibitem{Bj�rk Blix Landen}
Bj\"{o}rk, T., Blix, M. and Land\'{e}n, C. (2006). On finite dimensional realizations for the term structure of futures prices, \emph{International Journal of Theoretical and Applied Finance}, 9/3, 281 -- 314

\bibitem{Carmona Tehranchi}
Carmona, R. and Tehranchi, M. (2006). \emph{Interest Rate Models: an Infinite Dimensional Stochastic Analysis Perspective}, Springer Finance, Berlin: Springer


\bibitem{Cox Ingersoll Ross 81}
Cox, J. C., Ingersoll, J. E. and Ross, S. A. (1981). The relation between forward prices and futures prices, \emph{Journal of Financial Economics}, 9/4, 321 -- 346

\bibitem{DePrato Zabczyk}
De Prato, G. and Zabczyk, J. (1992). \emph{Stochastic Equations in Infinite Dimensions}, Cambridge University Press

\bibitem{Ekeland Taflin}
Ekeland, I. and Taflin, E. \emph{A Theory of Bond Portfolios}, 2005, \emph{The Annals of Applied Probability}, 15/2, 1260 -- 1305


\bibitem{Filipovic}
Filipovi\'{c}, D. (2001). \emph{Consistency Problems for Heath-Jarrow-Morton Interest Rate Models}, Springer

\bibitem{Filipovic Teichmann 1}
Filipovi\'{c}, D. and Teichmann, J. (2003). Existence of invariant manifolds for stochastic equations in infinite dimensions,
\emph{Journal of Functional Analysis}, 197, 398 -- 432

\bibitem{Filipovic Teichmann 2}
Filipovi\'{c}, D. and Teichmann, J. (2003). Regularity of finite-dimensional realizations for evolution equations,
\emph{Journal of Functional Analysis}, 197, 433 -- 446

\bibitem{Fleming Soner}
Fleming, W. H., and Soner, M. (2006). \emph{Controlled Markov Processes and Viscosity Solutions}, 2nd edn. Springer

\bibitem{Gawarecki Mandrekar}
Gawarecki, L. and Mandrekar, V. (2010). \emph{Stochastic Differential Equations in Infinite Dimensions}, Springer

\bibitem{G}
Geman, H. (2005). {\it Commodities and Commodity Derivatives}, Wiley-Finance.

\bibitem{GS}
Gibson, R. and Schwartz, E. S. (1990). Stochastic convenience yield and the pricing of
oil contingent claims, {\it Journal of Finance}, 45/3, 959 -- 976.

\bibitem{Heath Jarrow Morton}
Heath, D., Jarrow, R. and Morton, A. (1992). Bond pricing and the term structure of interest rates: a new methodology of contingent claims valuation, \emph{Econometrica}, 60/1, 77 -- 105

\bibitem{Karatzas Shreve}
Karatzas, I. annd Shreve, S. (1998). \emph{Methods of Mathematical Finance}, Springer

\bibitem{KO}
Koekebakker, S. and Ollmar, F. (2005). Forward curve dynamics in the Nordic electricity market. {\it Managerial Finance}, 31/6, 74  --95.


\bibitem{Ringer Tehranchi}
Ringer, N. and Tehranchi, M. \emph{Optimal Portfolio Choice in the Bond Market}, 2006, \emph{Finance and Stochastics}, 10, 553 -- 573


\bibitem{Tappe}
Tappe, S. (2010). An alternative approach on the existence of affine realizations for HJM term structure models, \emph{Proceedings of the Royal Society: Series A}, 466, 3033 -- 3060


\end{thebibliography}
\end{document}